\documentclass[submission,copyright,creativecommons]{eptcs}
\providecommand{\event}{GandAlf 2019} 
\usepackage[T1]{fontenc}
\usepackage{graphicx, algorithm, algorithmic}
\usepackage[utf8]{inputenc}
\usepackage{tikz}
\usetikzlibrary{shapes.geometric,shapes.misc}
\usetikzlibrary{calc,shapes,arrows,arrows.meta,shapes.misc,decorations.pathreplacing}
\usetikzlibrary{positioning,calc}
\usetikzlibrary{automata, positioning}

\usepackage{tikz-qtree}
\usepackage{subcaption}
\usepackage{lipsum}
\usepackage{listings}
\usepackage{adjustbox}
\usepackage{amsmath,amssymb,amsthm}
\usepackage{hyperref}
\usepackage{ textcomp }
\usepackage{ wasysym }

\newtheorem{proposition}{Proposition}

\newtheorem{theorem}{Theorem}
\newtheorem{lemma}{Lemma}
\newtheorem{corollary}{Corollary}

\usepackage{xcolor}

\definecolor{dkgreen}{rgb}{0,0.6,0}
\definecolor{gray}{rgb}{0.5,0.5,0.5}
\definecolor{mauve}{rgb}{0.58,0,0.82}

\lstdefinestyle{myScalastyle}{
	frame=tb,
	language=scala,
	aboveskip=3mm,
	belowskip=3mm,
	showstringspaces=false,
	columns=flexible,
	basicstyle={\small\ttfamily},
	numbers=none,
	numberstyle=\tiny\color{gray},
	keywordstyle=\color{blue},
	commentstyle=\color{dkgreen},
	stringstyle=\color{mauve},
	frame=single,
	breaklines=true,
	breakatwhitespace=true,
	tabsize=3,
	numbers=left,
	deletekeywords={for},
	otherkeywords={function,==,=, Let}
}

\title{On the Order Type of Scattered Context-Free Orderings\thanks{Ministry of Human Capacities, Hungary grant 20391-3/2018/FEKUSTRAT is acknowledged. Szabolcs Iv\'an was supported by the J\'anos Bolyai Scholarship of the Hungarian Academy of Sciences. Kitti Gelle was supported by the \'UNKP-19-3-SZTE-86 New National Excellence Program of the Ministry of Human Capacities.}}
\author{Kitti Gelle
	\institute{University of Szeged, Hungary}
	\email{kgelle@inf.u-szeged.hu}
	\and
	Szabolcs Iv\'an
	\institute{University of Szeged, Hungary}
	\email{szabivan@inf.u-szeged.hu}
}
\def\titlerunning{On the Order Type of Scattered Context-Free Orderings}
\def\authorrunning{Kitti Gelle \& Szabolcs Iv\'an}		
\begin{document}
	\maketitle
\begin{abstract}
	We show that if a context-free grammar generates a language whose lexicographic ordering is 
	well-ordered of type less than $\omega^2$, then its order type is effectively computable.
	
%
%
\end{abstract}
		
\section{Introduction}
If an alphabet $\Sigma$ is equipped by a linear order $<$, this order can be extended to the lexicographic
ordering $<_\ell$ on $\Sigma^*$ as $u<_\ell v$ if and only if either $u$ is a proper prefix of $v$ or
$u=xay$ and $v=xbz$ for some $x,y,z\in\Sigma^*$ and letters $a<b$. So any language $L\subseteq \Sigma^*$
can be viewed as a linear ordering $(L,<_\ell)$. Since $\{a,b\}^*$ contains the dense ordering
$(aa+bb)^*ab$ and every countable linear ordering can be embedded into any countably infinite dense ordering,
every countable linear ordering is isomorphic to one of the form $(L,<_\ell)$ for some language $L\subseteq\{a,b\}^*$.
A linear ordering (or an order type) is called \emph{regular} or \emph{context-free}
if it is isomorphic to the linear ordering (or, is the order type) of some language of the appropriate type.
It is known~\cite{DBLP:journals/fuin/BloomE10} that an ordinal is regular if and only if it is less than $\omega^\omega$
and is context-free if and only if it is less than $\omega^{\omega^\omega}$. Also, the Hausdorff rank~\cite{rosenstein}
of any scattered regular (context-free, resp.) ordering is less than $\omega$ ($\omega^\omega$, resp)~\cite{ITA_1980__14_2_131_0,10.1007/978-3-642-29344-3_25}.

It is known~\cite{GelleIvanTCS} that the order type of a well-ordered language generated by a prefix grammar (i.e. in which
each nonterminal generates a prefix-free language)
is computable,
thus the isomorphism problem of context-free ordinals is decidable if the ordinals in question are given as the lexicograpic
ordering of \emph{prefix} grammars.
Also, the isomorphism problem of regular orderings is decidable as well~\cite{DBLP:journals/ita/Thomas86,BLOOM200555}.
At the other hand, it is undecidable for a context-free grammar whether it generates a dense language,
hence the isomorphism problem of context-free orderings in general is undecidable~\cite{ESIK2011107}.

Algorithms that work for the well-ordered case can in many cases be ``tweaked'' somehow
to make them work for the scattered case as well:
e.g. it is decidable whether $(L,<_\ell)$ is well-ordered or scattered~\cite{10.1007/978-3-642-22321-1_19}
and the two algorithms are quite similar.

In this paper we continue to explore the boundary of decidability of the isomorphism problem of context-free orderings.
We show that if the order type $o(L)$ of a context-free language $L$ is known to have the form $\omega\times k+n$ for some integers $k$ and $n$, then
$k$ and $n$ can be effectively computed. The main building block for proving this is a decision procedure for
solving $o(L(X))\mathop{=}\limits^{?}\omega$ for each nonterminal $X$, and a recursive algorithm that terminates
for languages of order type less than $\omega^2$.

\section{Notation}
A \emph{linear ordering} is a pair $(Q,<)$, where $Q$ is some set and the $<$ is a transitive, antisymmetric and connex 
(that is, for each $x,y\in Q$ exactly one of $x<y$, $y<x$ or $x=y$ holds) binary relation on $Q$.
The pair $(Q,<)$ is also written simply $Q$ if the ordering is clear from the context.
A (necessarily injective) function $h: Q_1\to Q_2$, where $(Q_1,<_1)$ and $(Q_2,<_2)$ are some linear orderings,
is called an \emph{(order) embedding}
if for each $x,y\in Q_1$, $x<_1 y$ implies $h(x) <_2 h(y)$.
If $h$ is also surjective, $h$ is an \emph{isomorphism}, in which case the two orderings are \emph{isomorphic}.
An isomorphism class is called an \emph{order type}. The order type of the linear ordering $Q$ is denoted by $o(Q)$.

For example, the class of all linear orderings contain all the finite linear orderings
and the orderings of the integers ($\mathbb{Z}$), the positive integers ($\mathbb{N}$) and the negative integers ($\mathbb{N}_{-}$) whose order type is denoted $\zeta$, $\omega$ and $-\omega$ respectively.
Order types of the finite sets are denoted by their cardinality, and $[n]$ denotes $\{1,\ldots,n\}$ for each $n\geq 0$, ordered in the standard way.

The ordered sum $\sum_{x\in Q} Q_x$, where $Q$ is some linear ordering and for each $x \in Q$, $Q_x$ is a linear ordering,
is defined as the ordering with domain $\{(x,q):x\in Q,q\in Q_x\}$ and ordering relation
$(x,q)<(y,p)$ if and only if either $x<y$, or $x=y$ and $q<p$ in the respective $Q_x$. 
If each $Q_x$ has the same order type $o_1$ and $Q$ has order type $o_2$, then the above sum has order type $o_1\times o_2$.
If $Q=[2]$, then the sum is usally written as $Q_1+Q_2$.

If $(Q,<)$ is a linear ordering and $Q'\subseteq Q$, we also write $(Q',<)$ for the subordering of $(Q,<)$, that is, to ease notation we also use $<$ for the restriction of $<$ to $Q'$.

A linear ordering $(Q,<)$ is called \emph{dense} if it has at least two elements and for each $x,y\in Q$ where $x<y$ there exists a $z\in Q$ such that $x<z<y$.
A linear ordering is \emph{scattered} if no dense ordering can be embedded into it.
It is well-known that every scattered sum of scattered linear orderings is scattered, and any finite union of scattered linear orderings is scattered.
A linear ordering is called a \emph{well-ordering} if it has no subordering of type $-\omega$. Clearly, any well-ordering is scattered. Since isomorphism preserves well-orderedness
or scatteredness, we can call an order type well-ordered or scattered as well,
or say that an order type embeds into another.
The well-ordered order types are called \emph{ordinals}.
For any set $\Omega$ of ordinals, $(\Omega,<)$ is well-ordered by the relation $o_1<o_2~\Leftrightarrow~\hbox{``}o_1\hbox{ can be embedded injectively into }o_2$ but not vice versa''.
The principle of well-founded induction can be formulated as follows. Assume $P$ is a property of ordinals such that for any ordinal $o$, if $P$ holds for all ordinals
smaller than $o$, then $P$ holds for $o$. Then $P$ holds for all the ordinals.

For standard notions and useful facts about linear orderings see e.g.~\cite{rosenstein} or~\cite{jalex}.

Hausdorff classified the countable scattered linear orderings with respect to their rank.
We will use the definition of the Hausdorff rank from \cite{10.1007/978-3-642-29344-3_25},
which slightly differs from the original one (in which $H_0$ contains only the empty ordering and the singletons,
and the classes $H_\alpha$ are not required to be closed under finite sum, see e.g.~\cite{rosenstein}).
For each countable ordinal $\alpha$, we define the class $H_\alpha$ of countable linear orderings as follows.
$H_0$ consists of all finite linear orderings, and
when $\alpha> 0$ is a countable ordinal,
then $H_\alpha$ is the least class of linear orderings closed under finite ordered sum and isomorphism
which contains all linear orderings of the form $\sum_{i\in\mathbb{Z}} Q_i$,
where each $Q_i$ is in $H_{\beta_i}$ for some $\beta_i < \alpha$.

By Hausdorff's theorem, a countable linear order $Q$ is scattered if and only if it belongs to $H_\alpha$ for some countable ordinal $\alpha$.
The \emph{rank} $r(Q)$ of a countable scattered linear ordering is the least ordinal $\alpha$ with $Q \in H_\alpha$.

As an example, $\omega$, $\zeta$, $-\omega$ and $\omega+\zeta$ or any finite sum of the form $\mathop\sum\limits_{i\in[n]}o_i$ with $o_i\in\{\omega, -\omega,1\}$ for each $i\in[n]$
each have rank $1$ while $(\omega+\zeta)\times\omega$ has rank $2$.

Let $\Sigma$ be an alphabet (a finite nonempty set) and let $\Sigma^*$ ($\Sigma^+$, resp) stand for the set of all (all nonempty, resp)
finite words over $\Sigma$, $\varepsilon$ for the empty word, $|u|$ for the length of the word $u$, $u\cdot v$ or simply $uv$ for the concatenation of $u$ and $v$. A \emph{language} is an arbitrary subset $L$ of $\Sigma^*$. 
We assume that each alphabet is equipped by some (total) linear order.
Two (strict) partial orderings, the strict ordering $<_s$ and the prefix ordering $<_p$ are defined over $\Sigma^*$ as
follows:
\begin{itemize}
	\item $u <_s v$ if and only if $u = u_1au_2$ and $v = u_1bv_2$ for some words $u_1, u_2, v_2\in\Sigma^*$ and letters $a < b$,
	\item $u <_p v$ if and only if $v = uw$ for some nonempty word $w\in\Sigma^*$.
\end{itemize} 
The union of these partial orderings is the lexicographical ordering $<_\ell = <_s \cup <_p$.
We call the language $L$ well-ordered or scattered,
if $(L,<_\ell)$ has the appropriate property and we define the rank $r(L)$ of a scattered language $L$
as $r(L,<_\ell)$. The order type $o(L)$ of a language $L$ is the order type of $(L,<_\ell)$.
For example, if $a<b$, then $o\Bigl(\{a^kb:k\geq 0\}\Bigr)=-\omega$ and $o\Bigl(\{(bb)^ka:k\geq 0\}\Bigr)=\omega$.

When $\varrho$ is a relation over words (like $<_\ell$ or $<_s$), we write $K\varrho L$ if $u\varrho v$ for each word $u\in K$ and $v\in L$.

An \emph{$\omega$-word} over $\Sigma$ is an $\omega$-sequence $a_1a_2\ldots$ of letters $a_i\in\Sigma$. The set of all $\omega$-words over $\Sigma$ is denoted $\Sigma^\omega$.
The orderings $<_\ell$ and $<_p$ are extended to $\omega$-words. An $\omega$-word $w$ is called \emph{regular} if $w=uv^\omega=uvvvv\ldots$ for some finite words
$u\in\Sigma^*$ and $v\in\Sigma^+$. When $w$ is a (finite or $\omega$-) word over $\Sigma$ and $L\subseteq\Sigma^*$ is a language, then $L_{<w}$
stands for the language $\{u\in L:u<w\}$. Notions like $L_{\geq w}$, $L_{<_sw}$ are also used as well, with the analogous semantics.

A \emph{context-free grammar} is a tuple $G=(N,\Sigma,P,S)$, where $N$ is the alphabet of the \emph{nonterminal symbols}, $\Sigma$ is the alphabet of \emph{terminal symbols} (or \emph{letters}) which is disjoint from $N$, $S\in N$ is the \emph{start symbol} and $P$ is a finite set of \emph{productions} of the form $A\to\alpha$, where $A\in N$ and $\alpha$ is a \emph{sentential form}, that is, $\alpha = X_1X_2\ldots X_k$ for some $k\geq 0$ and $X_1,\ldots,X_k\in N\cup \Sigma$. The derivation relations $\Rightarrow$, $\Rightarrow_\ell$, $\Rightarrow^*$ and $\Rightarrow_\ell^*$ are defined as usual
(where the subscript $\ell$ stands for ``leftmost'').
The \emph{language generated by} a grammar $G$ is defined as $L(G) = \{u\in\Sigma^* ~|~ S\Rightarrow^*u\}$.
Languages generated by some context-free grammar are called \emph{context-free languages}. 
For any set $\Delta$ of sentential forms, the language generated by $\Delta$ is $L(\Delta) = \{u \in \Sigma^* ~|~ \alpha\Rightarrow^* u\hbox{ for some } \alpha\in\Delta\}$. 
As a shorthand, we define $o(\Delta)$ as $o(L(\Delta))$.
A language $L$ is \emph{prefix} (or prefix-free) if there are no words $u, v \in L$ with $u <_p v$.
A context-free grammar $G=(N,\Sigma,P,S)$ is called a \emph{prefix} grammar if $L(A)$ is a prefix language for each $A\in N$.
When $X,Y\in N\cup\Sigma$ are symbols of a grammar $G$, we write $Y\preceq X$ if $X\Rightarrow^*uYv$ for some words $u$ and $v$; $X\approx Y$ if $X\preceq Y$ and $Y\preceq X$ both hold;
and $Y\prec X$ if $Y\preceq X$ but not $X\preceq Y$. A production of the form $X\to X_1\ldots X_n$ with $X_i\prec X$
for each $i\in[n]$ is called an \emph{escaping production}.

A \emph{regular language} over $\Sigma$ is one which can be built up from the singleton languages $\{a\}$, $a\in\Sigma$
and the empty language $\emptyset$ with finitely many applications of taking (finite) union,
concatenation $KL=\{uv:u\in K,v\in L\}$ and iteration $K^*=\{u_1\ldots u_n:n\geq 0,u_i\in K\}$. 
For standard notions on regular and context-free languages the reader is referred to any standard textbook, such as~\cite{Hopcroft+Ullman/79/Introduction}.

Linear orderings which are isomorphic to the lexicographic ordering of some context-free (regular, resp.) language
are called \emph{context-free (regular, resp.) orderings}.

\section{If $o(L)<\omega^2$, then $o(L)$ is computable}
	
%
	
%

In this section we consider a context-free grammar $G=(N,\Sigma,P,S)$ which contains no left recursive nonterminals,
and generates a(n infinite) scattered language such that for each $X\in N$, $X$ is usable and $L(X)$ is an infinite language of nonempty words,
moreover, each nonterminal but possibly $S$ is recursive and there is no left recursive nonterminal (that is, $X\Rightarrow^+uXv$ implies $u\neq\varepsilon$).
Any context-free grammar can effectively be transformed into such a form, see e.g.~\cite{GelleIvanTCS}.

The section is broken into two parts: the first subsection contains some technical decidability lemmas, while the second
one contains the main result that if we know that $o(L)<\omega^2$ for a well-ordered context-free language $L$
(so that the Hausdorff-rank of $L$ is at most one), then $o(L)$ is effectively computable.
This computability is already known for so-called ordinal grammars which generate a well-ordered language such that
for each nonterminal $X$, $L(X)$ is a prefix language~\cite{GelleIvanTCS}.
However, this is a serious restriction and makes many proofs easier since if $K$ is a prefix language,
then $o(KL)=o(L)\times o(K)$ for any language $L$. This does not hold for arbitrary languages since e.g.
$o(a^*)=\omega$, $o(b)=1$ and $o(a^*b)=-\omega$, $o(a^*a^*)=\omega$, $o((ac)^*)=\omega$ and $o((ac)^*(b+ab))=\omega+(-\omega)$
so a more careful case analysis is required.
The reader is advised to skip
the first subsection at first read -- the proofs of the second part extensively refer to the lemmas of the first part.

\subsection{Some technical lemmas}

For an $\omega$-word $w$, let $\mathbf{Pref}(w)\subseteq\Sigma^*$ stand for the set of the finite prefixes of $w$.
For each $u=a_1\ldots a_k\in\Sigma^*$ and $v=b_1\ldots b_t\in\Sigma^+$
let $M_{u,v}$ denote the automaton (without specified final states)
depicted in Figure~\ref{fig-uv-automata}.
	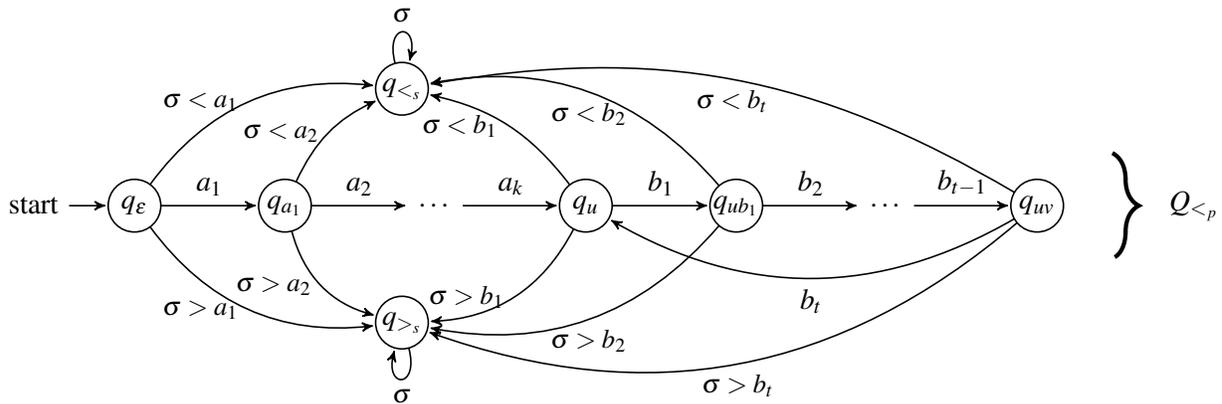
\begin{figure}[H]
	\centering
	\begin{tikzpicture}[-,>=stealth',shorten >=1pt,auto,node distance=2cm, semithick]
	\tikzstyle{every state}=[draw=black,text=black, inner sep = 0cm, outer sep = 0cm, minimum size = 0.7 cm]
	
	\node[initial,state] (e)               {$q_\epsilon$};
	\node[state]         (a1) [right of=e] {$q_{a_1}$};
	\node		         (d1) [right of=a1] {$\cdots$};
	\node[state]         (u)  [right of=d1] {$q_{u}$};
	\node[state]         (ub1) [right of=u]       {$q_{ub_1}$};
	\node		         (d2) [right of=ub1] {$\cdots$};
	\node[state]         (uv) [right of=d2]       {$q_{uv}$};
	
	\node[state]         (ss) [above right = 1.5cm of a1]       {$q_{<_s}$};
	\node[state]         (sb) [below right = 1.5cm of a1]       {$q_{>_s}$};
	
	\draw [decorate,decoration={brace,amplitude=10pt},yshift=0pt, ultra thick]
	(13,0.7) -- (13,-0.7) node [black,midway,xshift=0.6cm] {$Q_{<_p}$};
	
	\path[->]
	(e) edge node {$a_1$} (a1)
	(a1) edge node {$a_2$} (d1)
	(d1) edge node {$a_k$} (u)
	(u) edge node {$b_1$} (ub1)
	(ub1) edge node {$b_2$} (d2)
	(d2) edge node {$b_{t-1}$} (uv)
	(uv) edge[bend left=30] node {$b_t$} (u)

	(e) edge[bend left = 30, left] node {\small $\sigma<a_1$} (ss)
	(a1) edge[bend left = 20, left] node {\small $\sigma<a_2$} (ss)
	(u) edge[bend left = -20, left] node {\small $\sigma<b_1$} (ss)
	(ub1) edge[bend left = -30, below] node {\small $\sigma<b_2$} (ss)
	(uv) edge[bend left = -20,below] node {\small $\sigma<b_t$} (ss)
	
	(e) edge[bend left = -30, left] node {\small $\sigma>a_1$} (sb)
	(a1) edge[bend left = -30, left] node {\small $\sigma>a_2$} (sb)
	(u) edge[bend left = 30,left] node {\small $\sigma>b_1$} (sb)
	(ub1) edge[bend left = 30, below] node {\small $\sigma>b_2$} (sb)
	(uv) edge[bend left = 30,below] node {\small $\sigma>b_t$} (sb)
	(ss) edge[loop above] node {\small $\sigma$} (ss)
	(sb) edge[loop below] node {\small $\sigma$} (sb)
	;
	\end{tikzpicture}
	\caption{The automaton $M_{u,v}$}
	\label{fig-uv-automata}
\end{figure}

\begin{proposition}
	\label{prop-rajzos}
	For any words $u\in\Sigma^*$ and $v\in\Sigma^+$, the languages $\mathbf{Pref}(uv^\omega)$, $\{w\in\Sigma^*:w<_suv^\omega\}$ and $\{w\in\Sigma^*:uv^\omega<_sw\}$ are regular.
\end{proposition}
\begin{proof}
	Let $u=a_1\ldots a_k$ and $v=b_1\ldots b_t$ for the integers $k\geq 0$, $t>0$ and letters $a_i,b_j$ and consider the automaton $M_{u,v}$ given in Figure~\ref{fig-uv-automata}.
	Then, by setting $q_{<_s}$ ($q_{>_s}$, respectively) for the unique accepting state we recognize $\{w\in\Sigma^*:w<_suv^\omega\}$ ($\{w\in\Sigma^*:uv^\omega<_sw\}$, resp.),
	and by setting $Q_{<p}$ as the set of final states we recognize $\mathbf{Pref}(uv^\omega)$.
\end{proof}

\begin{lemma}
\label{lem-sequence}
	For each sentential form $\alpha$ with $L(\alpha)$ being infinite, we can generate a sequence $w_0,w_1,\ldots \in L(\alpha)$ and a regular word $w\in\Sigma^\omega$ satisfying one of the following cases:
	\begin{itemize}
		\item[i)] $w_1<_sw_2<_s\ldots$ and $w=\mathop\bigvee\limits_{i\geq 0}w_i$
		\item[ii)] $w_1>_sw_2>_s\ldots$ and $w=\mathop\bigwedge\limits_{i\geq 0}w_i$
		\item[iii)] $w_1<_pw_2<_p\ldots$ and $w=\mathop\bigvee\limits_{i\geq 0}w_i$
	\end{itemize}
\end{lemma}
\begin{proof}
By the pumping lemma of context-free languages, as $L(\alpha)$ is infinite, one can generate a word $u\in L(\alpha)$ 
and a partition $u = u_1u_2u_3u_4u_5$ such that $|u_2u_4|\geq 1$ and for each $n\geq 0$, the word $u_1u_2^nu_3u_4^nu_5$ is in $L(\alpha)$.  

Based on the relative order of the five subwords we consider the following cases:
\begin{enumerate}
	\item \textbf{There exists an $\boldsymbol{n_0}$ such that $\boldsymbol{u_3u_4^{n_0} <_s u_2u_3u_4^{n_0}}$}
	
	Let us define the sequence as $w_n = u_1u_2^{n_0+n}u_3u_4^{n_0+n}u_5$.
	Let us fix $n$ and let $m$ denote $m=n_0+n$.
	Here we get that $w_n <_s w_{n+1}$ if and only if $u_3u_4^{m}u_5 <_s u_2u_3u_4^{m+1}u_5$, which is true since 
	$u_3u_4^{n_0}<_su_2u_3u_4^{n_0}$ implies $u_3u_4^{n_0}x<_su_2u_3u_4^{n_0}y$ for any $x,y\in\Sigma^*$, thus in particular
	for $x=u_4^nu_5$ and $y=u_4^{n+1}u_5$ as well.
	
	Hence the type of the sequence is of i) and the supremum is $w = \mathop\bigvee\limits_{i\geq 0}w_i = u_1u_2^\omega$.
	(Observe that $u_2\neq\varepsilon$ as that could not satisfy $u_3u_4^{n_0}<_su_2u_3u_4^{n_0}$).
	
	\item \textbf{There exists an $\boldsymbol{n_0}$ such that $\boldsymbol{u_2u_3u_4^{n_0} <_s u_3u_4^{n_0}}$}
	
	Again, $u_2$ cannot be the empty word.	
	Similarly to the first case, let us define the sequence as $w_n = u_1u_2^{n_0+n}u_3u_4^{n_0+n}u_5$ and fix $n$ and let $m$ denote $m=n_0+n$.
	
	Here we get that $w_{n+1} <_s w_{n}$ if and only if $u_2u_3u_4^{m+1}u_5 <_s u_3u_4^{m}u_5$, which is true since 
	$u_2u_3u_4^{n_0}<_su_3u_4^{n_0}$ implies $u_2u_3u_4^{n_0}x<_su_3u_4^{n_0}y$ for any $x,y\in\Sigma^*$, thus in particular for $x=u_4^{n+1}u_5$ and $y=u_4^{n}u_5$ as well. So we get that the type of the sequence is ii) (with order type of $-\omega$) and the infimum is $w =\mathop\bigwedge\limits_{i\geq 0}w_i = u_1u_2^\omega$.
	
	\item \textbf{For each $\boldsymbol{n}$ it holds that $\boldsymbol{u_3u_4^n \leq_p u_2u_3u_4^{n}}$ and $\boldsymbol{u_4\neq \varepsilon}$}
		 
	 In this case $u_3u_4^\omega = u_2u_3u_4^\omega$.
	 
	 Let us fix $N = \left\lceil\frac{|u_2|}{|u_4|}\right\rceil+1$ and the sequence as $w_n=u_1u_2^{N+n}u_3u_4^{N+n}u_5$.
	 Furthermore, let $x\in\Sigma^*$ be the unique word with $u_3u_4^Nx=u_2u_3u_4^{N}$.
	 That is, $x$ is the unique suffix of $u_4^{N}$ of length $|u_2|$.
	 Hence, for any $n\geq N$ we also have $u_3u_4^nx=u_2u_3u_4^{n}$ for the same $x$ (as we know that $u_3u_4^n\leq_pu_2u_3u_4^{n}$, their length differ by $|u_2|$, and the latter word
	 ends with $u_4^{N}$).
	 
	 We have three subcases: 
	 \begin{enumerate}
	 	
	 	\item \textbf{It holds that $\boldsymbol{u_3u_4^Nu_5 <_s u_2u_3u_4^{N+1}u_5}$ }

	 	First observe that as $u_2u_3u_4^{N+1}u_5=u_3u_4^Nxu_4u_5$, the assumption of the subcase yields $u_5<_sxu_4u_5$.
	 	Then for each $m\geq N$ we have $u_3u_4^{m}u_5<_su_3u_4^mxu_4u_5=u_2u_3u_4^mu_4u_5=u_2u_3u_4^{m+1}u_5$,
	 	implying $u_1u_2^mu_3u_4^mu_5<_su_1u_2^{m+1}u_3u_4^{m+1}u_5$. So the sequence is of type i) and its supremum is either $u_1u_2^\omega$ (if $u_2$ is nonempty) or $u_1u_3u_4^\omega$ (otherwise).
	 	
	 	\item \textbf{It holds that $\boldsymbol{u_2u_3u_4^{N+1}u_5 <_s u_3u_4^{N}u_5}$ }
	 
		First since $u_2u_3u_4^{N+1}u_5$ can be written as $u_3u_4^Nxu_4u_5$, the assumption of the subcase yields $xu_4u_5 <_s u_5$.
		Then for each $m\geq N$ we have $u_3u_4^mxu_4u_5=u_2u_3u_4^mu_4u_5=u_2u_3u_4^{m+1}u_5 <_s u_3u_4^mu_5$, so we get a sequence of words such that $u_1u_2^{m+1}u_3u_4^{m+1}u_5 <_s u_1u_2^mu_3u_4^mu_5$. Hence, the sequence is of type ii) with the infimum of either $u_1u_2^\omega$ (if $u_2$ is nonempty) or $u_1u_3u_4^\omega$ (otherwise).
		
		\item \textbf{It holds that $\boldsymbol{u_3u_4^Nu_5 <_p u_2u_3u_4^{N+1}u_5}$ }
	 
		In this case for each $m\leq N$ we have $u_3u_4^mu_5 <_p u_2u_3u_4^{m+1}u_5$, so we get an ascending prefix chain. Hence the order type of this sequence of iii) is $\omega$ with the supremum of
		either $u_1u_2^\omega$ (if $u_2$ is nonempty) or $u_1u_3u_4^\omega$ (otherwise).		
	 \end{enumerate}
 	\item \textbf{It holds that $\boldsymbol{u_3 <_p u_2u_3}$ and $\boldsymbol{u_4=\varepsilon}$}
 	
 		Note that $u_2$ cannot be empty in this case. 
 		
 		
 		We have three subcases:
 		\begin{enumerate}
 			\item \textbf{It holds that} $\boldsymbol{u_3u_5<_su_2u_3u_5}$
 			
 			In this case for each $n\leq 0$ we have $u_1u_2^nu_3u_5 <_s u_1u_2^{n+1}u_3u_5$ iff $u_3u_5<_s u_2u_3u_5$ which is the assumption of this subcase. So we get that the sequence type is i) and the supremum is $w = u_1u_2^\omega$.
 			
 			\item \textbf{It holds that} $\boldsymbol{u_2u_3u_5<_su_3u_5}$
 			
 			Here, similarly to the previous case for each $n\leq 0$ we have $u_1u_2^{n+1}u_3u_5 <_s u_1u_2^nu_3u_5$ iff $u_2u_3u_5 <_s u_3u_5$, which is implied by the assumption. Hence we have an infinite descending chain with the sequence type of ii) and infimum $w=u_1u_2^\omega$. 
 			
 			\item \textbf{It holds that} $\boldsymbol{u_3u_5<_p u_2u_3u_5}$
 			
 			In the last case since $u_3u_5<_p u_2u_3u_5$, for each $n\leq 0$ we have $u_1u_2^nu_3u_5 <_p u_1u_2^{n+1}u_3u_5$ which is a prefix chain with sequence type of iii) and supremum $w = u_1u_2^\omega$.
 			
 		\end{enumerate}
\end{enumerate}
Observe that it is also decidable which (sub)case applies: first we check for the condition of Case 4 (which is clearly decidable). Then, if that condition does not hold, we check whether $u_3u_4^\omega=u_2u_3u_4^\omega$ holds.
As equality of regular words is decidable, this can be done, and if they are the same, then we again have three sub-conditions concerning finite words. Otherwise, if $u_3u_4^\omega\neq u_2u_3u_4^\omega$, then
either $u_3u_4^\omega<_s u_2u_3u_4^\omega$, that is, $u_3u_4^{n_0}<_s u_2u_3u_4^{n_0}$ for some $n_0$, or $u_2u_3u_4^\omega<_s u_3u_4^\omega$, in which case $u_2u_3u_4^{n_0}<_su_3u_4^{n_0}$ for some $n_0$.
But since we know that one of these two cases has to hold, we only have to iterate through all the integers $n$ and compare $u_3u_4^n$ with $u_2u_3u_4^n$ and eventually there will be an $n$ for which these two become comparable by $<_s$.
(A more efficient algorithm also exists, e.g. by analyzing the direct product automaton $M_{u_3,u_4}\times M_{u_2u_3,u_4}$.)
\end{proof}

We recall the following characterizations of those context-free grammars generating a scattered (or well-ordered) language from~\cite{DBLP:conf/fics/BloomE09}:
\begin{theorem}[\cite{DBLP:conf/fics/BloomE09}]
\label{thm-esik}
	Assume $G=(N,\Sigma,P,S)$ is a context-free grammar such that each nonterminal is usable, $\varepsilon$-free and there are no left-recursive nonterminals.
	Then
	\begin{itemize}
		\item $L(G)$ is scattered if and only if for each recursive nonterminal $X$ there exists a word $u_X\in\Sigma^+$ such that whenever $X\Rightarrow^+uX\alpha$ for some $u\in\Sigma^*$, $\alpha\in(N\cup\Sigma)^*$, then $u\in u_X^+$.
		\item If $L(G)$ is scattered and $X\approx X'$ are recursive nonterminals, then there exists a word $u_{X,X'}<_p u_X$ such that whenever $X\Rightarrow^+uX'\alpha$ for some $u\in\Sigma^*$,
			$\alpha\in(N\cup\Sigma)^*$, then $u\in u_X^*u_{X,X'}$.
		\item $L(G)$ is well-ordered if and only if it is scattered and for each recursive nonterminal $X$, $L(X)<_\ell u_X^\omega$.
	\end{itemize}
	Moreover, for each $X,X'$ the words $u_X$ and $u_{X,X'}$ are effectively computable and it is decidable whether $L(X)$ is scattered, or well-ordered.
\end{theorem}

\begin{proposition}
\label{prop-sup-ux}
	If $L(X)$ is well-ordered for the recursive nonterminal $X$, then $\bigvee L(X)=u_X^\omega$.
\end{proposition}
\begin{proof}
	By Theorem~\ref{thm-esik}, $L(X)<_\ell u_X^\omega$. Since $X$ is recursive and all the nonterminals are usable,
	there exists some derivation of the form $X\Rightarrow uXv$ for some words $u,v\in\Sigma^*$ and so $u=u_X^m$ for some $m>0$.
	Since $X$ is usable, there exists some word $w\in L(X)$ and so for each $n>0$, the word $u_X^{m\cdot n}wv^{n}$ is in $L(X)$ and is still upperbounded by $u_X^\omega$.
	As the supremum of these words is $u_X^\omega$, we got the claimed result.
\end{proof}

We call an infinite language $L\subseteq\Sigma^*$ a \emph{prefix chain} if for each $u,v\in L$, either $u\leq_pv$ or $v\leq_p u$,
that is, $L$ is totally ordered by the prefix relation, or equivalently, $L\subseteq\mathbf{Pref}(w)$ for some $\omega$-word $w$.
Clearly, any language $L$ is either a prefix chain or contains two words $u,v$ with $u<_sv$.
Note that if $L$ is a prefix chain, then $o(L)=\omega$.

\begin{lemma}
\label{lem-strict-pair}
	It is decidable for each nonterminal $X$ whether $L(X)$ is a prefix chain.
\end{lemma}
\begin{proof}
	By Lemma~\ref{lem-sequence} we can effectively generate an infinite sequence $w_0,w_1,\ldots$, either ascending or descending, belonging to $L(X)$ along with its limit, which is of the form $uv^\omega$ for some
	$u\in\Sigma^*$, $v\in\Sigma^+$.
	Now if the sequence is either a $>_s$-chain or a $<_s$-chain, then $L(X)$ cannot be a prefix chain.
	
	Otherwise, the sequence itself is a prefix chain and its limit is $uv^\omega$, hence the whole language $L(X)$ is a prefix chain if and only if $L(X)\subseteq \mathbf{Pref}(uv^\omega)$ which can be effectively
	decided since $\mathbf{Pref}(uv^\omega)$ is a regular language.
\end{proof}

\begin{lemma}
\label{lem-omega-supremum}
	If $L$ is a context-free language with $o(L)=\omega$, then $\bigvee L$ is a computable regular word.
\end{lemma}
\begin{proof}
	Applying Lemma~\ref{lem-sequence} we can generate a (necessarily increasing) sequence $w_0<w_1<\ldots$ of words belonging to $L$ along with their supremum $uv^\omega$.
	Since the order type of $L$ is also $\omega$, its supremum has to coincide by $uv^\omega$.
\end{proof}

\begin{lemma}
\label{lem-strict-bigger-than-omega}
	Let $X$ be a nonterminal such that $L(X)$ is not a prefix chain and $\alpha$ be a sentential form with $L(\alpha)$ being infinite.
	Then $o(X\alpha)$ is an infinite order type different from $\omega$.
\end{lemma}
\begin{proof}
	Since $L(X)$ is not a prefix chain and is infinite, there exists $u,v\in L(X)$ with $u<_sv$. Then $uL(\alpha)<_svy$ for any member $y$ of $L(\alpha)$, hence each such $vy$ has infinitely many lower bounds in $L(X\alpha)$,
	thus $o(X\alpha)$ cannot be $\omega$.
\end{proof}

The last lemma of the subsection is a bit technical:

\begin{lemma}
\label{lem-l1l2-automata-decidable}
	If $L_1\subseteq\mathbf{Pref}(uv^\omega)$ is a context-free prefix chain with order type $\omega$ for some words $u\in\Sigma^*$ and $v\in\Sigma^+$
	and $L_2\subseteq\Sigma^*$ is a context-free language with order type $\omega$,
	then it is decidable whether there exists some $w_1\in L_1$, $u'\in\Sigma^*$ and $a\in\Sigma$ with
	$w_1u'<_puv^\omega$, $w_1u'a<_suv^\omega$ and $u'a<_p\bigvee L_2$.
\end{lemma}
\begin{proof}
	Let us write $u=a_1\ldots a_k$ and $v=b_1\ldots b_t$ and consider the automaton $M_{u,v}$ of Figure~\ref{fig-uv-automata}.
	For each state $q$ of $M_{u,v}$, let $L_{u,v}(q)$ stand for the (regular) language $\{w\in\Sigma^*:q_\varepsilon\cdot w=q\}$.
	
	Let $Q_1\subseteq Q_{<_p}$ be the set of those states $q$ for which $L_1\cap L_{u,v}(q)$ is nonempty. Since each $L_{u,v}(q)$ is regular,
	$Q_1$ is computable, moreover, $q\in Q_1$ if and only if $q_{\varepsilon}\cdot w=q$ for some $w\in L_1$.
	
	Now by Lemma~\ref{lem-omega-supremum} we can compute the regular word $u_2v_2^\omega=\bigvee L_2$ and consider the direct product automaton
	$M=M_{u,v}\times M_{u_2,v_2}$ where in $M_{u_2,v_2}$ we use the primed version $q'$ of each state $q$.
	
	We claim that there exists words $w_1,u'$ and a letter $a$ satisfying the conditions of the lemma if and only
	if a state of the form $(q_{<_s},q')$ is reachable from a state $(p,q'_{\varepsilon})$ in $M$ for some $q'\in Q'_{<_p}$
	and $p\in Q_1$.
	
	Indeed: assume $(p,q'_\varepsilon)\cdot w=(q_{<_s},q')$ for such states: let us choose $p,q'$ and $w$ so that $|w|$ is the shortest possible.
	Since $p\in Q_1$ and $q_{<_s}\notin Q_1$, $w=u'a$ for some word $u'\in\Sigma^*$ and $a\in\Sigma$.
	Then, $q'_\varepsilon\cdot u'\in Q'_{<_p}$, since both $q'_{<_s}$ and $q'_{>_s}$ are trap states in $M_{u_2,v_2}$.
	Since $w$ is a shortest possible word and $q_{<_s}$, $q_{>_s}$ are trap states in $M_{u,v}$, we get that $p\cdot u'\in Q_{<_p}$.
	Since $p\in Q_1$, there is some word $w_1\in L_1$
	with $q_\varepsilon\cdot w_1=p$. Thus, this choice of $w_1,u'$ and $a$ satisfies the conditions of the lemma.
	
	And similarly, given $w_1\in L_1$, $u\in\Sigma^*$ and $a\in\Sigma$ satisfying the conditions we can define $p=q_\varepsilon\cdot w_1$,
    $q'=q'_\varepsilon\cdot u'a$.
\end{proof}
\subsection{The main decision procedures}
In this part we flesh out the ``top-level'' results leading to the aforementioned computability result: that the order type of well-ordered context-free languages
with Hausdorff-rank at most $1$ is computable.

The main building block is the result that it is decidable for any context-free language $L$ whether $o(L)=\omega$ holds.

\begin{lemma}
\label{lem-alpha-order}
	If $\alpha=X_1X_2$ is a sentential form and for each $1\leq i\leq 2$ we know whether $o(X_i)=\omega$ holds or not,
	then it is also effectively computable whether $o(\alpha)$ is $\omega$.
\end{lemma}
\begin{proof}
	Clearly, if $o(X_1)$ or $o(X_2)$ is some infinite order type different from $\omega$, then
	$o(\alpha)$ is also such an order type and we can stop. Also, if $X_1\in\Sigma$, then $o(\alpha)=o(X_2)$ and we are done.
	We can also decide whether $L(X_1X_2)$ is well-ordered and if not, it cannot be $\omega$ and we can stop.	
	
	So we can assume that $X_1\in N$ and thus $L(X_1)$ is infinite and hence $o(X_1)=\omega$ by assumption. Let $L_1$, $L_2$
	and $L$ respectively stand for $L(X_1)$, $L(X_2)$ and $L(X_1X_2)$.

	If $X_2\in\Sigma$, then we have several subcases:
	
			\begin{enumerate}
				\item If there exists a word $u\in L_1$ and some letter $a<X_2$ with $ua$ being a prefix of infinitely many words in $L_1$,
					then $uX_2$ is strictly larger than each of these words $v$, and so $vX_2<_suX_2$ as well, thus $o(L)$
					cannot be $\omega$ but some other infinite order type (as $uX_2\in L$ is preceded by infinitely many members of $L$).					
				\item Otherwise, let $u\in L_1$.
					It suffices to show that there are only finitely many elements in $L$ which are smaller than $uX_2$. 
					Assume $vX_2\in L$ is so that $vX_2<_puX_2$ then $v<_pu$ as well, and as $u$ has only finitely many proper prefixes,
					we have that there can only be a finite number of such words $v$. Now if $vX_2<_suX_2$, then either $v<_\ell u$
					(that's again a finite number of possibilities, as $o(L_1)=\omega$ implies that any word $u\in L_1$ has only
					a finite number of lower bounds in $L_1$) or $u<_\ell v$. Thus, in this case (as $vX_2<_suX_2$ rules out the possibility
					of $u<_sv$) it has to hold that $u<_pv$ and $v=uax$ for some $a<X_2$.
					There are only finitely many possible choices for such letters $a<X_2$ and by assumption (see the condition of the previous subcase),
					for each such letter, $ua$ can be a prefix of only finitely many words $v\in L_1$.
					
					Thus, $uX_2$ is larger than only a finite number of members of $L$ for each $u\in L_1$, and $o(L)=\omega$ in this case.
			\end{enumerate}
			We still have to show that it is decidable which of the two cases holds.
			Observe that if $ua$ is a prefix of infinitely many words in $L_1$ for some word $u\in L_1$ and letter $a<X_2$,
			then no word $w\in L_1$ can satisfy $ua<_s w$ as then the order type of $L_1$ could not be $\omega$.
			Thus, $ua<_p\bigvee L_1$ in this case for some word $u\in L_1$ and letter $a<X_2$.
			On the other hand, if $ua<_p\bigvee L_1$ for some word $u\in L_1$ and letter $a<X_2$, then 
			for any $w\in L_1$ with $ua<_\ell w$ we cannot have $ua<_s w$ since in that case $ua<_s\bigvee L_1$ would hold since $w<_\ell\bigvee L_1$.
			Hence, whenever $ua<_\ell w$ for some word $w\in L_1$, then $ua$ is a prefix of $w$. Since the order type of $L_1$ is assumed to be $\omega$,
			and $ua$ is a prefix of $\bigvee L_1$, there has to be an infinite number of such words $w$.
			
			Thus, if $X_2\in\Sigma$, then $o(L)$ is not $\omega$ if and only if $ua<_p\bigvee L_1$ for some $u\in L_1$ and $a<X_2$.
			This condition is decidable: $\bigvee L_1$ is a computable regular word $u_1v_1^\omega$ by Lemma~\ref{lem-omega-supremum}
			and we only have to check whether the language $L_1a\cap\mathbf{Pref}(u_1v_1^\omega)$ is nonempty for some letter $a<X_2$
			and the latter is a regular language by Proposition~\ref{prop-rajzos}.
			
		If $X_2\in N$, and thus $o(X_2)=\omega$, then we again have several cases:
			\begin{enumerate}
				\item[3.] If there exist words $u,v\in L_1$ with $u<_sv$, then by Lemma~\ref{lem-strict-bigger-than-omega} we get $o(L)$ is
					some infinite order type different than $\omega$.
				\item[4.] Otherwise, $L_1$ is an infinite prefix chain, that is, $L_1\subseteq\mathbf{Pref}(uv^\omega)$ for some words $u\in\Sigma^*$,
					$v\in\Sigma^+$.
					
					We have several subcases.
					\begin{enumerate}
						\item Assume $\bigvee L_1<_\ell\bigvee L$. Since both are $\omega$-words, we have $<_s$ here.
							Thus there exists some $w=w_1w_2\in L$, $w_1\in L_1$, $w_2\in L_2$ with $\bigvee L_1<_sw$.
							Since $L_1$ is an infinite prefix chain, there exists some $w_1'\in L_1$, $w_1'<_p\bigvee L_1$ and $|w_1'|>|w|$,
							yielding $w_1'<_s w$. So $w_1'L_2<_sw$, thus $w\in L$ has an infinite number of lower bounds in $L$
							and $o(L)\neq\omega$ in this subcase.
						\item Assume there exists some $w_1\in L_1$ such that $w_1\cdot \bigvee L_2<_\ell \bigvee L_1$. Again, both being $\omega$-words
							this has to be a $<_s$ relation. This means that there exists some $w_1'\in L_1$ with $w_1\cdot \bigvee L_2<_sw_1'$
							and so $w_1\cdot L_2<_sw_1'y$ for any member $y$ of $L_2$, thus again, $o(L)\neq\omega$ in this case.
						\item Assume none of the previous conditions hold: $\bigvee L\leq_\ell \bigvee L_1$
							(hence $\bigvee L_1L_2=\bigvee L_1=uv^\omega$ as well) and
							for each $w_1\in L_1$ we have $\bigvee L_1\leq_\ell w_1\cdot\bigvee L_2$.
																				
							We claim that in this case $o(L)=\omega$ if and only if for each $w_1\in L_1$ and $w<_puv^\omega$ with $w_1<_pw$
							there exist only finitely many words $w_2\in L_2$ such that $w_1w_2<_sw$.
							
							For one direction, assume the latter condition holds. It suffices to show that for each $w<_puv^\omega$ there exist
							only finitely many many words $w_1\in L_1$, $w_2\in L_2$ with $w_1w_2<_\ell w$, since (as the supremum of these words
							is $\bigvee L_1L_2=uv^\omega$) this yields that each prefix of $o(L)$ is finite, thus $o(L)=\omega$.
							So let $w_1w_2<_\ell w$. Since $w_1\in L_1$ and $L_2\subseteq\mathbf{Pref}(uv^\omega)$, and $w<_puv^\omega$,
							we either have $w_1<_pw$ or $w\leq_p w_1$. The latter would contradict to $w_1w_2<_\ell w$, hence we have $w_1<_pw$.
							Thus, there are only finitely many options for choosing such a word $w_1\in L_1$.
							Clearly, for each fixed $w_1\in L_1$ there are only finitely many options for choosing words $w_2$ with $w_1w_2<_pw$ and by the condition
							there are only finitely many words $w_2\in L_2$ with $w_1w_2<_sw$, hence in total, there are only finitely many words
							in $L_1L_2$ preceding $w$, showing $o(L)=\omega$.
							
							For the other direction, assume the latter condition does \emph{not} hold.
							Then there exists some $w_1\in L_1$, $w<_puv^\omega$ with $w_1<_pw$ such that $w_1w_2<_sw$ for infinitely many words $w_2\in L_2$.
							In this case we can write $w_2=w_2'ax$ and $w=w_1w_2'by$ uniquely for some letters $a<b$ and words $w_2',x,y$.
							Since there are only finitely many options for the fixed words $w$ and $w_1$ to choose $w_2'$, $b$ and $a$,
							for some pair $a<b$ of letters and word $w_2'$ with $w_1w_2'b\leq_pw$ there are infinitely many words $w_2\in L_2$ such that
							$w_2'a\leq_pw_2$. Let $L_2'\subseteq L_2$ denote the (infinite) set of these words and let $w_1'\in L_1$ be 
							some word in $L_1$ with $|w_1'|\geq |w_1w_2'b|$. Since $L_1$ is an infinite prefix chain, such a word $w_1'$ exists
							and $w_1w_2'b\leq_pw_1'$, thus $w_1L_2'<_sw_1'$, and so $w_1L_2'<_sw_1'y$ for an arbitrary member $y$ of $L_2$,
							and so $w_1'y$ is preceded by infinitely many words in $L$, yielding $o(L)\neq \omega$.
					\end{enumerate}
					We still have to show that the condition of Subcase (c) is decidable. We claim that the condition does \emph{not} hold
					if and only if there exists some $w_1\in L_1$ and words $u'\in\Sigma^*$, $a\in\Sigma$ with 
					$w_1u'<_puv^\omega$, $w_1u'a<_suv^\omega$ and $u'a<_p\bigvee L_2$. Indeed, if $w_1,u',a$ are such objects,
					then there is a unique letter $b\in\Sigma$ with $w_1u'b<_puv^\omega$. Now we can choose $w=w_1u'b$ as the condition
					$u'a<_p\bigvee L_2$ implies the existence of infinitely many words $w_2\in L$ with $u'a\leq_p w_2$. The other direction
					is already treated in the proof of Subcase (c).
					
					The condition in this form is decidable due to Lemma~\ref{lem-l1l2-automata-decidable}.
			\end{enumerate}
		As we covered all the possible scenarios, and in each case we got decidability, we proved the lemma.
\end{proof}
\begin{corollary}
\label{cor-alpha-computable}
	Assume $\alpha=X_1\ldots X_n$ is some sentential form where for each $X_i$ we know whether $o(X_i)=\omega$ holds.
	Then we can decide whether $o(\alpha)=\omega$ holds.
\end{corollary}
\begin{proof}
	We can assume that $\alpha\notin\Sigma^*$ (otherwise $o(\alpha)=1$ and we can stop).
	Using the standard construction of introducing fresh nonterminals $Y_1,\ldots Y_{n-1}$ and 
	productions $Y_1\to X_1Y_2$, $Y_2\to X_2Y_3$,\ldots, $Y_{n-1}\to X_{n-1}X_n$ we can successively decide for $Y_{n-1}, Y_{n-2},\ldots,Y_1$
	whether $o(Y_i)=\omega$; if for any of them we have that $o(Y_i)$ is some other infinite order type, then so is $o(Y_1)=o(\alpha)$,
	otherwise $o(\alpha)=\omega$.
\end{proof}

\begin{theorem}
\label{thm-omega-decidable}
	It is decidable for each recursive nonterminal $X$ whether $o(X)=\omega$ holds.
\end{theorem}
\begin{proof}
	By our assumptions of $G$, $L(X)$ is infinite.
	In the first step, we decide whether $L(X)$ is well-ordered. If not, then $o(X)$ is clearly not $\omega$ and we can stop.
	
	From now on, we know that $L(X)$ is well-ordered.
	Next, we decide whether $L(X)$ contains two words $u,v$ with $u<_sv$. If not, then $L(X)$ is a prefix chain and then
	$o(X)=\omega$. So we can assume that there exist members $u_0<_sv_0$ of $L(X)$. By $u_0<_sv_0<_\ell \bigvee L(X)=u_X^\omega$ (by Proposition~\ref{prop-sup-ux}),
	we get that $u_0<_su_X^\omega$.
	
	Now, we check whether there exists a sentential form $\beta$ containing at least one nonterminal with $X\Rightarrow^+ uX\beta$ for some $u\in u_X^+$.
	(Such a $\beta$ exists if and only if there is a production of the form $X'\to \alpha X''\gamma$ with $X'\approx X''\approx X$ and with $\gamma$ containing at least one nonterminal.)
	If so, then $L(\beta)$ is infinite, moreover, $uu_0L(\beta)<_suv_0w$ for any member $w$ of $L(\beta)$. Since $uu_0L(\beta)\subseteq L(X)$ and $uv_0w\in L(X)$, we get that
	$L(X)$ has some element preceded by an infinite number of other members of $L(X)$, hence $o(X)$ cannot be $\omega$.
	
	Hence we can assume that whenever $X\Rightarrow^+ uX\beta$, then $\beta\in\Sigma^*$ and $u\in u_X^+$.
	
	By induction, we can decide for each nonterminal $Y\prec X$ whether $o(Y)=\omega$ or not. Since into $\omega$ no other
	infinite order type can be embedded, if there exists some $Y\prec X$ with $o(Y)\neq\omega$, we can conclude $o(X)\neq\omega$ as well. So we can assume $o(Y)=\omega$ for each $Y\prec X$.
	
	Now let us consider an escaping production $X'\to\alpha$ with $X'\approx X$ such that $\alpha$ contains at least one nonterminal, thus $L(\alpha)$ is infinite.
	By Corollary~\ref{cor-alpha-computable}, we can effectively decide whether $o(\alpha)=\omega$ or not -- if not, then $o(X)$ again cannot be $\omega$ and we can stop.
	Hence, we can assume that for any such escaping production, $o(\alpha)=\omega$ as well. 
	By Lemma~\ref{lem-sequence}, we can generate a sequence $w_0,w_1,\ldots$ belonging to $L(\alpha)$ which is either an ascending or a descending chain.
	Since $L(\alpha)$ can be embedded into $L(X')$, it has to be well-ordered as well, ruling out the possibility of being a descending chain.
	Thus, we can compute the supremum $w=\bigvee w_i$ of the sequence as well, which, as $o(\alpha)=\omega$, has to be $\bigvee L(\alpha)$.
	
	Clearly, $w=\bigvee L(\alpha)\leq \bigvee L(X')=u_{X'}^\omega$ as $L(\alpha)\subseteq L(X')$,
	so either $w=u_{X'}^\omega$ or $w<_su_{X'}^\omega$. If $w<_su_{X'}^\omega$, that is, $w<_su_{X'}^N$ for some $N>0$, then, as $L(X')$ contains some word $x$ beginning with $u_{X'}^N$,
	we get that $L(\alpha)$ is an infinite subset of $L(X')$ strictly smaller than some member of $L(X')$, hence $o(X')$, hence also $o(X)$ are also greater than $\omega$ in that case as well and we can stop.
	
	Hence, we can assume that for each escaping production $X'\to\alpha$ with $X'\approx X$ and with $\alpha$ containing some nonterminal
	we have $\bigvee L(\alpha)=u_{X'}^\omega$, and $o(\alpha)=\omega$. As any finite union of linear orderings of order type $\omega$ is still $\omega$ if the suprema of the orderings coincide,
	we get that if $\alpha_1,\ldots,\alpha_t$ are all the alternatives for some nonterminal $X'\approx X$, then $o\Bigl(\mathop\bigcup\limits_{i\in[t]}L(\alpha_i)\Bigr)=\omega$
	with supremum $u_{X'}^\omega$.
	
	We claim that in this case $o(X)=\omega$. Since $\bigvee L(X)=u_X^\omega$, it suffices to show for each fixed $w<_pu_X^\omega$ that $L(X)_{<_\ell w}$ is finite.
	Since there is only a finite number of prefixes of $w$, this is equivalent to state the finiteness of $L(X)_{<_s w}$.
	Each word $w'\in L(X)$ can be factored as $w'=u_X^nu_{X,X'}zv$ for some integer $n\geq 0$, nonterminal $X'\approx X$ and words $z,v\in\Sigma^*$ such that
	$z\in L(\alpha)$ for some sentential form $\alpha$ for which $X'\to\alpha$ is an escaping production. If $w'<_s w$, then we have an upper bound
	for $n$, which in turn places an upper bound for $|v|$ since there are no left-recursive nonterminals.
	Hence, there are only finitely many possibilities for choosing $n$, $X'$, $\alpha$ and $v$
	and it suffices to see that for each such choice, the number of possible words $w'$ is finite.
	
	So let us write $w$ as $w=u_X^nu_{X,X'}w_1$. The condition $w'=u_X^nu_{X,X'}zv<_su_X^nu_{X,X'}w_1=w$ is equivalent to
	$zv<_sw_1$ for the fixed words $v$ and $w_1$. Let us assume that there are infinitely many such words $w'<_s w$ under the chosen values of $n$, $X'$, $\alpha$ and $v$.
	This entails $zv<_sw_1$ for an infinite number of words $z\in L(\alpha)$. Now as $zv<_sw_1$ can happen if either $z<_sw_1$ or $z<_pw_1$ and $w_1=zw_1'$ for some $v<_sw_1'$
	and this latter case can hold only for a finite number of words $z$ (as there are only finitely many prefixes of $w_1$), there has to be
	an infinite number of words $z$ in $L(\alpha)$ with $z<_sw_1$ for the finite prefix $w_1$ of $u_{X'}^\omega$.
	Let $L'$ be the (infinite) set of these words $z$. Then we have $\bigvee L'\leq_\ell w_1$ but since $L'$ is infinite and of order type $\omega$ (as $o(\alpha)=\omega$ as well),
	$\bigvee L'$ is an $\omega$-word, thus $\bigvee L'<_sw_1<_pu_{X'}^\omega$. Moreover, as $o(\alpha)=\omega$ and $L'$ is an infinite subset of $L(\alpha)$,
	it has to be the case that $\bigvee L'=\bigvee L(\alpha)$ so $\bigvee L(\alpha)<_su_{X'}^\omega$ which is a contradiction as we know that $\bigvee L(\alpha)=u_{X'}^\omega$.
	Hence there can be only a finite number of possible words $z$, showing our claim that $o(X)=\omega$.
\end{proof}

Combining Theorem~\ref{thm-omega-decidable} and Corollary~\ref{cor-alpha-computable} we get:
\begin{corollary}
	It is decidable for each sentential form $\alpha$ whether $o(\alpha)=\omega$ holds.
	Also, it is decidable whether $o(L(G))=\omega$ holds.
\end{corollary}
\begin{proof}
	The first statement is a direct consequence of Theorem~\ref{thm-omega-decidable} and Corollary~\ref{cor-alpha-computable}.
	Now if the start symbol $S$ of $G$ is recursive, then $o(L(G))\mathop{=}\limits^{?}\omega$ is decidable by
	Theorem~\ref{thm-omega-decidable}. Otherwise, if $\alpha_1,\ldots,\alpha_k$ are the alternatives of $S$, then
	$o(L(G))=\omega$ if and only if $o(\alpha_i)=\omega$ for each $i\in[k]$ generating an infinite language and moreover,
	$\bigvee L(\alpha_i)=\bigvee\mathop\bigcup\limits_{i\in[k]}L(\alpha_i)$ for each $i\in[k]$ with $L(\alpha_i)$ being
	infinite. These conditions are decidable by Corollary~\ref{thm-omega-decidable} and Lemma~\ref{lem-omega-supremum}.
\end{proof}

Now we are ready to show the main result of the paper:
\begin{theorem}
\label{thm-omega-omega-computable}
	There exists an algorithm which awaits a context-free grammar $G=(N,\Sigma,P,S)$ generating a well-ordered language
	and if $o(L(G))<\omega^2$, then returns $o(L(G))$ (otherwise enters an infinite recursion).
\end{theorem}
\begin{proof}
	We claim that the following algorithm correctly computes $o(L(G))$ and terminates whenever $o(L(G))<\omega^2$:
	\begin{enumerate}
		\item If $L(G)$ is finite, then return its size.
		\item If $o(L(G))=\omega$, then return $\omega$.
		\item Generate a sequence $w_0<w_1<\ldots$ of members of $L(G)$ and their supremum $uv^\omega$.
		\item If $L(G)_{\geq uv^\omega}$ is nonempty, then compute $o_1=o(L(G)_{<uv^\omega})$ and $o_2=o(L(G)_{>uv^\omega})$
			recursively and return $o_1+o_2$.
		\item Otherwise, $uv^\omega=\bigvee L(G)$. For each $n\geq 0$ in increasing order, test whether $o(L(G)_{>w_n})=\omega$ holds.
			If so, then compute $o_1=o(L(G)_{\leq w_n})$ recursively and return $o_1+\omega$. Otherwise increase $n$.
	\end{enumerate}
	Let us consider an example run of the above algorithm first, then we prove its correctness and conditional
	termination.
	
	Let the language $L$ be $a^*+b^na^n+c$. It has the order type $\omega\times 2 + 1$, since the lexicograpical ordering of the words of $L$ looks like $\varepsilon \leq_\ell a \leq_\ell aa \leq_\ell aaa \leq_\ell \ldots \leq_\ell ba \leq_\ell bbaa \leq_\ell bbbaaa \leq_\ell \ldots \leq_\ell c$.
	
	First, we check the finiteness of the language. Since $L$ is not finite, we check whether it has the order type $\omega$. Since it is not $\omega$, we generate an infinite sequence of words of $L$ and its supremum. Let us say we generate the sequence $ba<b^2a^2<\ldots<b^na^n<\ldots$ with its supremum $b^\omega$. Now we check that whether the language $L$ contains any word which is lexicographically greater than this supremum.
	Since $L_{\geq b^\omega} = c$ is nonempty, we compute $o(a^*+b^na^n)$ and $o(c)$ recursively (and at the end, we return their sum).
	
	So we compute $o(K)$ for $K=a^*+b^na^n$, which should be $\omega\times 2$.
	Since it is not finite and also not $\omega$, we try to cut it and generate a sequence again.	
	Say we generate the sequence $a < aa < ba < b^2a^2 < \ldots < b^na^n < \ldots$ of its members, with its supremum $b^\omega$ (note that the algorithm never generates a sequence like this since we use the pumping lemma to generate the words, so we use it just to explain how the algorithm works in a case like this).

	The language $K_{>b^\omega}$ is empty, the $\omega$-word, $b^\omega$ is the supremum of $K$ as well, so we have to find another cut point. To achieve this, we iterate through the sequence from $i=1$ and check whether the language containing the greater words than the word $w_i$:
	\begin{itemize}
		\item Cutting with the first word of the sequence we get that $o(K_{>a}) = o(aa(a^*)+b^na^n)$ is not $\omega$.
		(It's $\omega\times 2$.) So we increase the index and try again with the next word.
		\item If we use the second word, we have the same case: $o(K_{>aa}) = o(aaa(a^*)+b^na^n)$ is not $\omega$. (It's still $\omega\times 2$.) We increase the index.
		\item With the third word of the sequence we get that $o(K_{>ba}) = o(bbb^naaa^n)$ is $\omega$.
	\end{itemize}
	
	Now we have a new valid cut point, so we compute $o(K_{\leq ba}) = o(a^*+ba)$ recursively and return $o(a^*+ba)+\omega$. 
	So we have a recursive case again with the language $M = a^*+ba$. Since it is not finite, and its order type
	is still not $\omega$ -- it is actually $\omega+1$ -- we generate a sequence again. Say we generate the 
	sequence $a < aa < aaa < aaaa < \ldots$ with the supremum $a^\omega$. (In fact, it's guaranteed we generate
	some subsequence of this, with the same supremum.)
	Since $M_{\geq a^\omega} = ba$, we compute recursively $o(a^*)$ and $o(ba)$. Since the language $a^*$ has the order type $\omega$ we return it. Also, as $ba$ is finite, we return its size, $1$. So we get that $o(M) = \omega + 1$. 
	So on the previous level we get that $o(M)+\omega = \omega + 1 + \omega = \omega\times 2$.
	
	Finally, we compute $o(c)$. As $c$ is finite, we return its size and get $o(c) = 1$. At the top level we get back $\omega\times 2+1$ as the order type of the language $L$.

	Finishing this example, we first prove the conditional termination.
	The first three steps clearly terminate according to Theorem~\ref{thm-omega-decidable} and Lemma~\ref{lem-sequence}
	and the fact that it is decidable whether $L(G)$ is finite and if so, its members can be effectively enumerated. 
	
	For the supremum $uv^\omega$ of the words $w_0,w_1,\ldots\in L(G)$, we know that $L(G)_{<uv^\omega}$ is infinite (thus is at least $\omega$).
	Hence, $o(L(G))=o(L(G)_{<uv^\omega})+o(L(G)_{>uv^\omega})$ and the first summand is an infinite ordinal. If $L(G)_{\geq uv^\omega}$ is nonempty
	(which can be decided as well), then the second summand is also a nonzero ordinal. Thus, the first summand is guaranteed to be strictly smaller
	than $o(L(G))$ and since $o(L(G))$ is assumed to be smaller than $\omega^2$, and $o(L(G))=o_1+o_2$ for some $o_1\geq\omega$,
	we get that $o_2<o(L(G))$ as well. Thus, both recursive calls have an argument with a strictly less order type than $o(L(G))$, so these
	calls will eventually terminate, by well-founded induction.
	
	Finally, in step 5 we know that $\bigvee w_n=\bigvee L(G)$ and that $o(L(G))>\omega$ (as $o(L(G))\leq\omega$ is handled by the first two steps),
	thus as $o(L(G))<\omega^2$ by assumption, we have $o(L(G)_{>w_n})=\omega$ for some $n$. Thus, the iteration of step 5 eventually
	finds such an $n$ and terminates (as in that case $o(L(G)_{\leq w_n})$ is also less than $o(L(G))$, thus we can apply well-founded induction again
	for the recursive call).
	
	Correctness is clear since for steps 1 and 2 there is nothing to prove, step 3 cannot return anything and both of steps 4 and 5
	create a cut of $L(G)$ of the form $L(G)=L(G)_{\leq w}+L(G)_{>w}$ for some suitable (finite or infinite) word $w$, computes the order types
	of the two languages (which have strictly smaller order type, so applying well-founded induction we get their order type gets computed correctly)
	and returns their sum, which is correct.
\end{proof}

\section{Conclusion}
We showed that if $L(G)$ is known to be well-ordered with Hausdorff-rank at most one,
then $o(L(G))$ is computable. We strongly suspect that this result holds for the scattered case as well:
in fact, if it is decidable for a recursive nonterminal $X$ whether $o(X)=-\omega$ holds, then (by an algorithm
very similar to the one we gave for the well-ordered case) we can show that the order type of any scattered context-free
language of rank at most one is effectively computable. Note that $o(L_1L_2)$ can be $-\omega$ even if $o(L_1)=\omega$
as e.g. $o(a^*)=\omega$ but $o(a^*b)=-\omega$.

An open problem from~\cite{ESIK2011107} is the decidability status of the isomorphism problem of
\emph{deterministic} context-free
orderings (which form a proper subset of the unambiguous context-free ones).
The lexicographic orderings of deterministic context-free languages are called \emph{algebraic orderings} there
as they are exactly those isomorphic to the linear ordering of the leaves of an algebraic tree~\cite{doi:10.1142/S0129054111008155} in the sense of~\cite{COURCELLE198395}.

We do not know whether the isomorphism problem of scattered context-free orderings of rank $2$ is decidable:
by a standard reduction from the PCP problem one can construct a context-free grammar $G$ for an instance
$(u_1,v_1),\ldots,(u_n,v_n)$ of the PCP so that $L(G)$ is the union of the three languages
\[\{i_t\ldots i_1\cent u_{i_1}\ldots u_{i_t}\cent(aa)^k:t\geq 1,i_1,\ldots,i_t\in[n],k\geq 0\},\]
\[\{i_t\ldots i_1\cent v_{i_1}\ldots v_{i_t}\cent a(aa)^k:t\geq 1,i_1,\ldots,i_t\in[n],k\geq 0\}\]
and $\{i_t\ldots i_1\cent \$ a^k:t\geq 1,i_1,\ldots,i_t\in[n],k\geq 0\}$. Then, for each fixed $t\geq 1$
and $i_1,\ldots,i_t\in[n]$ the language $(i_t\ldots i_1)^{-1}L(G)$ has order type $\omega+-\omega$
if $i_1,\ldots,i_t$ is a solution of the given PCP instance and $\omega+\omega+-\omega$ otherwise,
hence $o(L(G))~=~\mathop\sum\limits_{t\geq 1,~i_1,\ldots,i_t\in[n]}o\Bigl((i_t\ldots i_1)^{-1}L(G)\Bigr)$
has the order type $\mathop\sum\limits_{t\geq 1,~i_1,\ldots,i_t\in[n]}(\omega+\omega+-\omega)$
if and only if the instance has no solution, where the tuples $(i_t,\ldots,i_1)$ are ordered lexicographically.
This latter sum is a quasi-dense sum of scattered orderings though,
so it does not prove undecidability of the isomorphism problem of scattered context-free orderings immediately,
but we conjecture that the problem is indeed undecidable.

Also, both context-free linear orderings in general and in the scattered case lack a characterization~\cite{DBLP:conf/fics/BloomE09}:
in fact, it is unclear which scattered orderings of rank two are context-free. (For rank one it is clear
as the rank one scattered order types are exactly the finite sums of natural numbers, $\omega$s and $-\omega$s
and these sums are all context-free, in fact, they are regular.)
\bibliography{biblio}{}
\bibliographystyle{eptcs}
\end{document}